\documentclass[12pt]{article}
\usepackage{amssymb}
\usepackage{amsmath}
\usepackage{color}
\usepackage{inputenc}
\usepackage[colorlinks,linkcolor=blue,citecolor=green]{hyperref}
\newtheorem{theorem}{Theorem}

\newtheorem{corollary}[theorem]{Corollary}

\newtheorem{example}[theorem]{Example}
\newtheorem{lemma}[theorem]{Lemma}

\newtheorem{proposition}[theorem]{Proposition}
\newtheorem{remark}[theorem]{Remark}
\newenvironment{proof}[1][Proof]{\textbf{#1.} }{\ \rule{0.5em}{0.5em}}

\allowdisplaybreaks[3]
\begin{document}
\title{\Large \textbf{
Quasilinear systems of first order PDEs\\
with nonlocal Hamiltonian structures}}
\author{\large Pierandrea Vergallo 
\\[3mm]
  \small  Department of Mathematical, Computer,\\
  \small Physical and Earth Sciences\\
  \small University of Messina, \\
  \small V.le F. Stagno D'Alcontres 31, I-98166 Messina, Italy\\
   \texttt{pierandrea.vergallo@unime.it}
}
\date{\small }
  
\maketitle

\begin{abstract}
\noindent In this paper we wonder whether a quasilinear system of PDEs of first order admits Hamiltonian formulation with local and nonlocal  operators. By using the theory of differential coverings, we find differential-geometric conditions necessary to write a given system with one of the three Hamiltonian operators investigated. 

\bigskip

\bigskip

\end{abstract}

\newpage 
\tableofcontents
\section{Introduction}
Hamiltonian formalism for systems of Partial Differential Equations is a consolidate tool in Integrable systems. In particular, it represents an important instrument for the study of nonlinear differential equations (e.g. \cite{DubKri,
  NovikovManakovPitaevskiiZakharov:TS}).  To find a Hamiltonian structure for a given system of PDEs is not immediate, but it is intrinsically connected to integrability.  This geometric approach to Hamiltonian formalism has been deeply investigated in the last fifty years, in particular for homogeneous operators.

\noindent The theory of homogeneous Hamiltonian operators was introduced by B. A. Dubrovin and S. P. Novikov \cite{DubrovinNovikov:PBHT, DN83}. These operators represent a natural tool when dealing with quasilinear systems of first order, especially in the homogeneous form
\begin{equation}\label{1}
u^i_t=V^i_j(u)u^j_x
\end{equation}
where $u^j=u^j(t,x)$ are field variables, $j=1,\dots , n$ depending on two independent variables $t,x$ and $V^i_j$ is a matrix whose entries are functions of $u^j$. Such systems  \eqref{1} are called \emph{hydrodynamic type systems}.

\noindent As an example, in \cite{tsarev85:_poiss_hamil} S. P. Tsarev  proved that systems of type \eqref{1} are strongly connected with first order homogeneous Hamiltonian operators
\begin{equation}\label{2}
A^{ij}=g^{ij}\partial_x+\Gamma^{ij}_ku^k_x,
\end{equation}
 also known as \emph{Dubrovin-Novikov operators}.

\noindent In the nondegenerate case  (i. e. $\det(g^{ij})\neq 0$) let us consider the following bracket associated to \eqref{2}
\begin{equation}\label{aaad}
\{F,G\}_A=\int{\frac{\delta F}{\delta u^i}A^{ij}\frac{\delta G}{\delta u^j}dx}=\int{\frac{\delta F}{\delta u^i}(g^{ij}\partial_x+\Gamma^{ij}_ku^k_x)\frac{\delta G}{\delta u^j}dx}.
\end{equation}

\noindent In \cite{DN83} B. A. Dubrovin and S. P. Novikov proved that  \eqref{aaad} is a Poisson bracket if and only if  $g^{ij}$ is a flat contravariant metric and $\Gamma^{ij}_k=-g^{is}\Gamma^j_{sk}$ are Christoffel symbols of the Levi-Civita connection of $g_{ij}$ (the inverse of $g^{ij}$).

\noindent Finally, in \cite{tsarev85:_poiss_hamil}, S. P. Tsarev shows that hydrodynamic type systems \eqref{1} admit the Hamiltonian formulation \begin{equation}
u^i_t=A^{ij}\frac{\delta H}{\delta u^j}=(g^{ij}\partial_x+\Gamma^{ij}_ku^k_x)\frac{\delta H}{\delta u^j}
\end{equation}
through a Dubrovin-Novikov operator \eqref{2} if and only if $g_{ik}V^k_j=g_{jk}V^k_i$ for every $i,j=1,\dots , n$ and $\nabla_iV^k_j=\nabla_jV^k_i$ for every $i,j,k=1,\dots ,n$.

\noindent In general, it not easy to establish this type of \emph{compatibilty conditions} between operators and systems of PDEs. A new approach to the problem of finding Hamiltonian formulations for hydrodynamic type systems through homogeneous Hamiltonian operators was firstly presented in \cite{FPV14, FerPavVit1} and \cite{PavVit1}, and finally consolidated in \cite{VerVit1} (see also \cite{VerVit2}). Necessary conditions of compatibility between hydrodynamic type systems and homogeneous operators can be found by using the theory of coverings due to  I.S. Krasil'shchik, P. Kersten and A. Verbovetsky in \cite{KerstenKrasilshchikVerbovetsky:HOpC} . First E. V. Ferapontov, M. V. Pavlov and R. Vitolo found a set of necessary conditions for third order homogeneous Hamiltonian operators \cite{FPV14} and then R. Vitolo and the author found analog conditions for homogeneous operators of second order  \cite{VerVit2}. These conditions have also projective algebraic geometry properties as proved in \cite{FerPavVit1} and the associated hydrodynamic systems possess a geometric meaning in the theory of projective line congruences, as shown by  S. I. Agafonov and E. V. Ferapontov (see for example \cite{AgaFer}).

\vspace{5mm}
\noindent This paper aims at presenting necessary conditions of compatibility for a larger class of systems of PDEs
which can be called \emph{non-homogeneous hydrodynamic type systems}
\begin{equation}\label{nhyd}
u^i_t=V^i_j(u)u^j_x+W^i(x,u).
\end{equation}
In particular, the main goal in what follows is to find conditions such that systems \eqref{nhyd}  admit a Hamiltonian formulation through nonlocal operators. First, compatibility conditions are proved for first order homogeneous Hamiltonian operators (as well as what was done for systems of type \eqref{1} in \cite{VerVit1}) and examples of such systems are presented. Then, the second operator studied is a non-local Ferapontov operator, i.e.
\begin{equation}
\tilde{A}^{ij}=g^{ij}\partial_x+\Gamma^{ij}_ku^k_x+f^i\partial_x^{-1}f^j,
\end{equation}
introduced by E. V. Ferapontov in \cite{ferapontov92:_non_hamil}.
The resulting conditions of compatibility turn out to be a useful tool when looking for isometry extensions of homogeneous quasilinear systems of PDEs. In particular, two examples of systems of PDEs in form \eqref{nhyd} are presented, the constant astigmatism equations and a bi-Hamiltonian system of dimension 3. 
\noindent The latter is a non-local Ferapontov-Mokhov operator:
\begin{equation}\label{c1}
B^{ij}=g^{ij}\partial_x+\Gamma^{ij}_ku^k_x+cu^i_x\partial_x^{-1}u^j_x,
\end{equation} 
where $c$ is a constant. These operators were studied by E. V. Ferapontov and O. I. Mokhov in \cite{MokFer12}.
An example of hydrodynamic type system possessing a Hamiltonian operator \eqref{c1} is the Chaplygin gas equation studied by Ferapontov and Mokhov and also analyzed in the present paper. 
 
 \vspace{3mm}
\noindent The paper is composed of three sections. In the first one the essential tools of tangent and cotangent differential coverings are presented for quasilinear systems of PDEs. Here, the reader can find the necessary condition for a system to admit Hamiltonian formulation: this is the starting point for the computations in what follows. The second section recalls the role of Dubrovin-Novikov operators for homogeneous systems and generalizes the results for non-homogeneous ones. The main section of the paper is the last one: the geometric conditions for quasilinear systems to admit  non-local operators are explictly found and discussed. Some examples of non-homogeneous systems of PDEs are computed with first-order homogeneous operators extended by isometries and Ferapontov-Mokhov operators. 

\vspace{5mm}
\noindent Symbolic computations arise in what follows several times, in particular when checking the compatibility conditions for the presented examples. Most of the computations in the present paper are performed by using  \texttt{Reduce} and \texttt{Maple}. For a symbolic approach to differential-geometric computations of systems and operators, the reader can use the package \texttt{CDE} available for \texttt{Reduce}. The package was implemented by R. Vitolo and represents a useful tool when dealing with Hamiltonian operators and jet spaces. In general, the use of symbolic computations for integrable systems and geometric structures has become a consolidated tool in the last years (see \cite{Vit2,Vit1}). A unified approch to this topic was presented by I. S. Krasil'shchik, A. Verbovetsky and R. Vitolo in \cite{KVV17}.

\section{Coverings and non-homogeneous quasilinear systems}
The method of cotangent covering was firstly presented in \cite{KerstenKrasilshchikVerbovetsky:HOpC} by Kersten, Krasil'shchik and Verbovetsky. Here we briefly recall how it works. Let us consider an evolutionary system
\begin{equation}\label{ev}
F^i=u^i_t-f(t,x,u^j,u^j_x,\dots )=0, \quad \quad i=1,\dots , n,
\end{equation}
where $(t,x)$ are the independent variables and $u^j$ are field variables. By definition,  the system \eqref{ev} admits a Hamiltonian formulation if there exist a differential operator $A^{ij}=a^{ij\sigma}D_\sigma$ and a functional $H=\int{h(u)\, dx}$ such that
\begin{equation}
u^i_t=A^{ij}\frac{\delta H}{\delta u^j}
\end{equation}
where $A$ is skew-adjoint and its Schouten brackets vanish ($[A,A]=0$). Hence $A$ is a \emph{Hamiltonian operator}. In this formalism, the functional $H$ is called the \emph{Hamiltonian} function of the system.

\noindent In this paper, non-homogeneous  hydrodynamic type system  are investigated:
\begin{equation}\label{sysnh}
u^i_t=V^i_ju^j_x+W^i
\end{equation}
where $V^i_j=V^i_j(\textbf{u})$ and $W^i=W^i(\textbf{u},x)$. This type of systems is form-invariant under transformations $\tilde{u}^i=\tilde{U}^i(\textbf{u})$. We will focus on the evolutionary system
\begin{equation}\label{hyd}
F^i: u^i_t-V^i_ju^j-W^i =0,\quad  \quad i=1,\dots , n,
\end{equation}
and we look for its Hamiltonian formalism.

\noindent In order to introduce the theory of differential coverings, it is necessary to recall that symmetries of \eqref{hyd} are vector functions $\varphi=\varphi^i$ such that $\ell_F(\varphi)=0$ when $F=0$, where $\ell_F$ is the Frech\'et derivative or the linearization of $\varphi$ with respect to $F$. Conservation laws are equivalence classes of 1-form $\omega=a dt+b dx$ that are closed modulo $F=0$ up to total divergencies. A conservation law is uniquely represented by generating functions $\psi_j=\delta b/\delta u^j$, which are called cosymmetries of the system. Such vector functions $\psi$ satisfy $\ell_F^*(\psi)=0$, where $\ell_F^*$ is the formal adjoint of $\ell_F$. 

\noindent For the system \eqref{hyd}, the operators are the following
\begin{align}
\ell_F(\varphi)^i&= D_t(\varphi^i)-(V^i_{j,l}u^j_x+W^i_{,l})\varphi^l-V^i_jD_x\varphi^j,\\
\ell_F^*(\psi)_i&=-D_t\psi_i+(V^k_{i,j}u^j_x-V^k_{j,i}u^j_x-W^k_{,i})\psi_k+V^k_iD_x\psi_k.
\end{align}

\noindent Here we used the notation $P_{,i}$ to indicate the partial derivative of $P$ with respect to $u^j$ and $P_{,x}$ the partial derivative with respect to the independet variable $x$.\\
\noindent  It is proved (\cite{Olver:ApLGDEq,KVV17}) that if $A$ is a Hamiltonian operator for $F=0$, then \begin{equation}\label{bivec}
\ell_F\circ A= A^*\circ \ell_F^*.
\end{equation}
It follows that if $\psi_k$ is a cosymmetry then $\ell_F(A^{ij}\psi_j)=0$, that is $\varphi^i=A^{ij}\psi_j$ is a symmetry of the system \eqref{hyd}. This implies that Hamiltonian operators map conserved quantities into symmetries. 

\noindent The approach introduced by Kersten, Krasil'shchik and Verbovetsky in \cite{KerstenKrasilshchikVerbovetsky:HOpC}  is the following. Let us introduce new variables $p_i$, such that we can associate $D_x\psi_i$ to $p_{i,x}$, $D^2_{x}\psi_i$ to $p_{i,xx}$ and so on. This correspondence allows us to associate to each differential operator $A^{ij}=a^{ij\sigma}D_\sigma\psi_j$ a linear vector function $A^{i}=a^{ij\sigma}p_{j,\sigma}$ where $\sigma$ identifies the derivation order with respect to $x$.

\noindent Now, let us introduce the \emph{cotangent covering} 
\begin{equation}
\mathcal{T}^*:\begin{cases}F=0\\\ell_F^*(\textbf{p})=0\end{cases}.
\end{equation}

\noindent For non-homogeneous hydrodynamic type systems it can be explicitly written as
\begin{equation}\label{a2}
\mathcal{T}^*:\begin{cases}u^i_t=V^i_ju^j_x+W^i\\p_{i,t}=(V^k_{i,j}u^j_x-V^k_{j,i}u^j_x-W^k_{,i})p_k+V^k_ip_{k,x}\end{cases}
\end{equation}

\noindent Let us introduce variables $q^i$ such that  we correspond $q^i$ for vector functions $\varphi^i$, $q^i_x$ to $D_x\varphi^i$ and so on. Analogously, we introduce the \emph{tangent covering}
\begin{equation}
\mathcal{T}:\begin{cases}F=0\\\ell_F(\textbf{q})=0\end{cases},
\end{equation} 
and for non-homogeneous hydrodynamic type systems \eqref{hyd}
\begin{equation}\label{a1}
\mathcal{T}:\begin{cases}u^i_t=V^i_ju^j_x+W^i\\
q^i_t=(V^i_{j,l}u^j_x+W^i_{,l})q^l+V^i_jq^j_x\end{cases}.
\end{equation}

\noindent Cotangent and tangent coverings are form-invariant under trasformations of type $\tilde{u}^i=\tilde{U}^i(\textbf{u})$.

\noindent We stress that the condition \eqref{bivec} can be considered as a necessary condition for a system $F=0$ to admit a Hamiltonian formulation through the Hamiltonian operator $A$. In general, by using the new setting of variables $\textbf{p}$ the following result holds true:

\begin{theorem}[\cite{KerstenKrasilshchikVerbovetsky:HOpC}]A linear vector function $A$ in total derivatives satisfies \eqref{bivec} if and only if the equation
\begin{equation}\label{nes}
\ell_F(A(\textbf{p}))=0
\end{equation}holds on the cotangent covering \eqref{a2}.\end{theorem}
The previous result is invariant in form under transformations $\tilde{u}^i=\tilde{U}^i(\textbf{u})$ if $A$ is also invariant. In particular, we focus on operators satisfying this condition.

 \noindent As stressed in \cite{VerVit1}, the condition \eqref{nes} does not introduce the Hamiltonian function $H=\int{h(u)\, dx}$ and represents an easier method to compute explicit \emph{necessary conditions} of compatibilty between a system \eqref{hyd} and an operator $A$. We emphasise that the previous condition does not guarantee that an operator $A$ is skew-symmetry or its Schouten brackets annihilate, then it does not ensure the hamiltonianity of the operator.

\section{Local operators}

Firstly, we investigate local structures for homogeneous and non-homogeneous hydrodynamic type systems. In order to do this, let us consider the first order homogeneous operators
\begin{equation}\label{dubnov}
A^{ij}=g^{ij}\partial_x+\Gamma^{ij}_ku^k_x
\end{equation}
where the tensor $g^{ij}$ is taken to be non-degenerate. Then,
\begin{theorem}[\cite{DN83}]A first order homogeneous operator is Hamiltonian if and only if 
\begin{enumerate}
\item if $g_{ij}=(g^{ij})^{-1}$, then $g_{ij}$ is a flat metric;
\item $\Gamma^{ij}_k=-g^{is}\Gamma^j_{sk}$, where $\Gamma^j_{sk}$ are Christoffel symbols for the metric $g_{ij}$.
\end{enumerate}
\end{theorem}
We obtain that the following hamiltonianity conditions must be satisfied:
\begin{align}
g^{ij}_{,k}&=\Gamma^{ij}_k+\Gamma^{ji}_k\\
R^{ij}_{lk}&=\Gamma^{ij}_{l,k}-\Gamma^{ij}_{k,l}+\Gamma^i_{ks}\Gamma^{sj}_l-\Gamma^j_{ks}\Gamma^{si}_l=0
\end{align}
where the last one is flatness of the metric given in coordinates by the Riemann curvature tensor.

\noindent We briefly recall that S. P. Tsarev found a set of necessary and sufficient conditions for a hydrodynamic type system to admit Dubrovin-Novikov operators:
\begin{lemma}A homogeneous system of first order PDEs 
\begin{equation}
u^i_t=V^i_j(u)u^j_x
\end{equation} admits a Dubrovin-Novikov operator \eqref{dubnov} if and only if there exists a flat non-degenerate metric $g_{ij}$ such that
\begin{itemize}
\item[(i)] $g_{is}V^s_j=g_{js}V^s_i$, 
\item[(ii)] $\nabla_iV^j_k=\nabla_jV^i_k$.
\end{itemize}
where $\nabla$ is the covariant derivative with respect to the metric $g_{ij}$.\end{lemma}

\noindent The aim of this section is to find similar conditions in the non-homogeneous case, i.e. when $W^i(x,u)$ is non-zero. In order to do this we will use the approach presented in the previous section. The conditions we will obtain are valid for more general systems, although they are only necessay and not sufficient.

\noindent Let us associate to  \eqref{dubnov} the linear vector function
\begin{equation}\label{eq:22}
A^i(\mathbf{p})=g^{ij}p_{j,x}+\Gamma^{ij}_ku^k_xp_j
\end{equation}

\noindent Therefore, using \eqref{a1} and \eqref{a2}, the linearization of the operator $A$ is computed:\footnote{The derivative $\partial_x$ is a total derivative as well as $\partial_t$.Whereas $W_x$ is the partial derivative with respect to the independent variable $x$.}
\begin{align}
\ell_F(A^i)&=\partial_tA^i-(V^i_{j,l}u^j_x+W^i_{,l})A^l-V^i_j\partial_xA^j\\
\begin{split}&=g^{ij}_lu^l_tp_{j,x}+g^{ij}p_{j,xt}+\Gamma^{ij}_{k,l}u^l_tu^k_xp_j+\Gamma^{ij}_ku^k_{xt}p_j+\Gamma^{ij}_ku^k_xp_{j,t}+\\
&\hphantom{ciao}-g^{lk}V^i_{j,l}u^j_xp_{k,x}-\Gamma^{lh}_kV^i_{j,l}u^j_xu^k_xp_h-W^i_{,l}g^{lk}p_{k,x}-W^i_{,l}\Gamma^{lh}_ku^k_xp_h\\
&\hphantom{ciao}-V^i_j(g^{jl}_ku^k_xp_{l,x}+g^{jl}p_{l,xx}+\Gamma^{jl}_{k,l}u^l_xu^k_xp_l+\Gamma^{jl}_{k}u^k_{xx}p_l+\Gamma^{jl}_ku^k_xp_{l,x})\\
&=g^{ij}_lV^l_ku^k_xp_{j,x}+g^{ij}_lW^lp_{j,x}+\\
&\hphantom{ciao}+g^{ij}V^k_{j,lh}u^h_xu^l_xp_k+g^{ij}V^k_{j,l}u^l_{xx}p_k+g^{ij}V^k_{j,l}u^l_xp_{k,x}+\\&\hphantom{ciao}\hphantom{ciao}-g^{ij}V^k_{l,jh}u^h_xu^l_xp_k-g^{ij}V^k_{l,j}u^l_{xx}p_k-g^{ij}V^k_{l,j}u^l_xp_{k,x}+\\
&\hphantom{ciao}-g^{ij}W^k_{,jx}p_k-g^{ij}W^k_{,jh}u^h_xp_k-g^{ij}W^k_{,j}p_{k,x}+\\
&\hphantom{ciao}+g^{ij}V^k_{j,l}u^l_{xx}p_{k,x}+g^{ij}V^k_{j}p_{k,xx}+\\
&\hphantom{ciao}+\Gamma^{ij}_{k,l}V^l_hu^h_xu^k_xp_j+\Gamma^{ij}_{k,l}W^lu^k_xp_j+\\
&\hphantom{ciao}+\Gamma^{ij}_kV^k_{l,h}u^h_xu^l_xp_j+\Gamma^{ij}_kV^k_lu^l_{xx}p_j+\Gamma^{ij}_kW^k_{,x}p_j+\Gamma^{ij}_kW^k_{,l}u^l_xp_j+\\
&\hphantom{ciao}+\Gamma^{ij}V^l_{j,h}u^k_xu^h_xp_l-\Gamma^{ij}_kV^l_{h,j}u^h_xu^k_xp_j-\Gamma^{ij}_kW^l_{,j}u^k_xp_l+\Gamma^{ij}_kV^l_ju^k_xp_{l,x}+\\
&\hphantom{ciao}-g^{lk}V^i_{j,l}u^j_xp_{k,x}-\Gamma^{lh}_kV^i_{j,l}u^j_xu^k_xp_h-W^i_{,l}g^{lk}p_{k,x}-W^i_{,l}\Gamma^{lh}_ku^k_xp_h\\
&\hphantom{ciao}-V^i_jg^{jl}_ku^k_xp_{l,x}-V^i_jg^{jl}p_{l,xx}+\\
&\hphantom{ciaociao}-V^i_j\Gamma^{jl}_{k,h}u^k_xu^h_xp_l-V^i_j\Gamma^{jl}_ku^k_{xx}p_l-V^i_j\Gamma^{jl}_ku^k_xp_{l,x}\end{split}
\end{align}

\noindent Let us now compute $\ell_F(A)=0$, by collecting for each undetermined variable 
\begin{equation*}\label{var}p_{l,xx},\quad u^l_{xx}p_k,\quad u^k_xp_{j,x},\quad p_{j,x},\quad u^k_xp_j,\quad  p_j,\end{equation*} then the following conditions must be satisfied:
\begin{subequations}\begin{align}
&\label{11a}-V^i_lg^{jl}+V^j_lg^{il}=0\\
&\label{13a}g^{ik}(V^j_{k,h}-V^j_{h,k})+\Gamma^{ij}_kV^k_h-\Gamma^{kj}_hV^i_k=0\\
&\label{12a} g^{ij}_lV^l_k+g^{il}(V^j_{l,k}-V^j_{k,l})+\Gamma^{il}_kV^j_l-g^{lj}V^i_{k,l}-g^{lj}_kV^i_l-\Gamma^{lj}_kV^i_l=0\\
&\label{14a}g^{ij}_lW^l-g^{il}W^j_{,l}-g^{lj}W^i_{,l}=0\\
&\label{15a}-g^{il}W^j_{,lk}+\Gamma^{ij}_{k,l}W^l+\Gamma^{ij}_lW^l_{,k}-\Gamma^{il}_kW^j_{,l}-\Gamma^{lj}_kW^i_{,l}=0\\
&\label{16a}-g^{il}W^j_{,xl}+\Gamma^{ij}_kW^j_x=0,
\end{align}\end{subequations}
plus another one which was proved  to be a differential consequence of the previous ones \cite{VerVit1}.

\begin{lemma}Condition \eqref{15a} is equivalent to $\nabla_k\nabla^iW^j=0$.\end{lemma}
\begin{proof}
 Consider that 
\begin{equation*}
\nabla^iW^j=g^{is}\nabla_sW^k=g^{is}W^k_{,s}-\Gamma^{ij}_sW^s
\end{equation*}
then, by using $g^{ij}_k=\Gamma^{ij}_k+\Gamma^{ji}_k$:
\begin{align}\begin{split}\nabla_k\nabla^iW^j&=(\nabla^iW^j)_k+\Gamma^{i}_{kl}\nabla^lW^j+\Gamma^j_{kl}\nabla^iW^l\\
&=\left(g^{il}W^j_{,l}-\Gamma^{ij}_lW^l\right)_k+\Gamma^i_{kl}\left(g^{ls}W^j_{,s}-\Gamma^{lj}_sW^s\right)\\
&\hphantom{ciao}+\Gamma^{j}_{kl}\left(g^{is}W^l_{,s}-\Gamma^{il}_sW^s\right)\\
&=g^{il}_kW^j_{,l}+g^{ij}W^k_{lk}-\Gamma^{ij}_lW^l_{,k}-\Gamma^{ij}_{l,k}W^l\\
&\hphantom{ciaociao}g^{ls}\Gamma^{i}_{kl}W^j_{,s}-\Gamma^i_{kl}\Gamma^{lj}_sW^s+\Gamma^j_{kl}g^{is}W^l_{,s}-\Gamma^j_{kl}\Gamma^{il}_sW^s\end{split}\\
\begin{split}\label{eq23}&=\Gamma^{il}_kW^j_{,l}+\Gamma^{li}_kW^j_{,l}+g^{il}W^j_{,lk}-\Gamma^{ij}_lW^l_{,k}-\Gamma^{ij}_{l,k}W^l-\Gamma^{si}_kW^j_{,s}\\
&\hphantom{ciaociao}-\Gamma^i_{kl}\Gamma^{lj}_sW^s+\Gamma^j_{kl}g^{is}W^l_{,s}-\Gamma^j_{kl}\Gamma^{il}_sW^s\end{split}
\end{align}
Note that two terms cancel out and the last one in \eqref{eq23} can be substituted by $-\Gamma^{j}_{kl}g^{il}_sW^s+\Gamma^j_{kl}\Gamma^{li}_sW^s$, therefore
\begin{align*}
\begin{split}\nabla_k\nabla^iW^j&=\Gamma^{il}_kW^j_{,l}+g^{il}W^j_{,lk}-\Gamma^{ij}_lW^l_{,k}-\Gamma^{ij}_{l,k}W^l\\
&\hphantom{ciao}-\Gamma^i_{kl}\Gamma^{lj}_sW^s+\Gamma^j_{kl}g^{is}W^l_{,s}-\Gamma^{j}_{kl}g^{il}_sW^s+\Gamma^j_{kl}\Gamma^{li}_sW^s.\end{split}
\end{align*}

\noindent Moreover, from the Hamiltonianity condition on \eqref{dubnov}, it follows that 
\begin{equation*}
0=R^{ij}_{lk}=\Gamma^{ij}_{l,k}-\Gamma^{ij}_{k,l}+\Gamma^i_{ks}\Gamma^{sj}_l-\Gamma^j_{ks}\Gamma^{si}_l,
\end{equation*}
and then
\begin{equation*}
\Gamma^{ij}_{l,k}-\Gamma^j_{kl}\Gamma^{li}_s=\Gamma^{ij}_{k,l}-\Gamma^i_{ks}\Gamma^{sj}_l.
\end{equation*}
We obtain that 
\begin{align}
\begin{split}\nabla_k\nabla^iW^j &=  \Gamma^{il}_kW^j_{,l}+g^{il}W^j_{,lk}-\Gamma^{ij}_lW^l_{,k}-\Gamma^{ij}_{k,l}W^l\\
&\hphantom{ciao}-\Gamma^i_{kl}\Gamma^{lj}_sW^s+\Gamma^j_{kl}g^{is}W^l_{,s}-\Gamma^{j}_{kl}g^{il}_sW^s+\Gamma^i_{kl}\Gamma^{lj}_sW^s.\end{split}
\end{align}
Now, by condition \eqref{14a} and by deleting two terms
\begin{equation*}
\nabla_k\nabla^iW^j=\Gamma^{il}_kW^j_{,l}+g^{il}W^j_{,lk}-\Gamma^{ij}_lW^l_{,k}-\Gamma^{ij}_{k,l}W^l
-\Gamma^j_{kl}g^{sl}W^i_{,s}.
\end{equation*}
Finally, $g^{sl}\Gamma_{kl}^jW^i_{,s}=-\Gamma^{sj}_kW^i_{,s}$ and the Lemma is proved.
\end{proof}

\vspace*{10mm}

\noindent Moreover, we can rewrite conditions \eqref{14a} and \eqref{16a} using a coordinate-free description (as  geometric conditions): 
\begin{equation}
\nabla^iW^j+\nabla^jW^i=-g^{ij}_lW^l+g^{il}W^j_{,l}+g^{lj}W^i_{,l}
\end{equation}and
\begin{equation}
\nabla^iW^j_{,x}=g^{il}W^j_{,xl}-\Gamma^{ij}_kW^k_x
\end{equation}
Collecting the previous results, we obtain the following
\begin{theorem}\label{thmloc}
Let us consider a local first-order homogeneous Hamiltonian operator and the associated linear vector function
\begin{equation*}
  A(\mathbf{p})^i = g^{ij}p_{j,x} + \Gamma^{ij}_k u^k_x p_j,
\end{equation*} if $F^i:  
u^i_t=V^i_j(u)u^j_x+W^i(x,u)$,
the following are equivalent:
\begin{itemize}
\item[(i)]$\ell_F(A^i(\mathbf{p}))=0$;
\item[(ii)] \begin{enumerate}
\item $\nabla^iV^j_k=\nabla^jV^i_k$;
\item $g^{ik}V^j_k=g^{jk}V^i_k$;
\item $\nabla^iW^j+\nabla^jW^i=0$;
\item $\nabla^iW^j_{,x}=0$;
\item $\nabla_k\nabla^iW^j=0$
\end{enumerate}\end{itemize}
\end{theorem}
B.A. Dubrovin and S.I.Novikov suggested to consider flat coordinates for a first-order homogeneous Hamiltonian operator, in which the operator takes the form 
\begin{equation}
A^{ij}=\eta^{ij}\partial_x
\end{equation}where $\eta^{ij}$ is a constant matrix. In these coordinates, the last three conditions in the previous theorem become easier:
\begin{align}
(3) \quad \eta^{il}W^j_{,l}+\eta^{jl}W^i_{,l}=0;\qquad 
(4)\quad  \eta^{il}W^j_{,xl}=0; \qquad 
(5) \quad -\eta^{ij}W^j_{,lk}=0. 
\end{align}
Explicitly solving them, we obtain \begin{equation}\label{e1}
W^i=a^i_ku^k+f^i(x),
\end{equation} where $a^j_k$ are arbitrary constants such that $\eta^{is}a_s^j=\eta^{js}a_s^i$ and $f^i(x)$ are arbitrary functions  $x$-dependent. In the particular case when $\eta$ is the antidiagonal unitary metric, $a^i_j=a^j_i$. 

\noindent In 2013, M.V. Pavlov and S.A. Zykov found a Bi-Hamiltonian structure for the constant astigmatism equation  \cite{PavZyk}
\begin{equation}\label{asty}
u_{tt}+\left(\frac{1}{u}\right)_{xx}+2=0
\end{equation} One of the structures analyzed is exactly of Dubrovin-Novikov type. Then, one can check that conditions in Theorem \ref{thmloc} are satisfied.
\begin{example}\label{examasti}
Let us introduce a new variable in order to write the equation \eqref{asty} as a quasilinear system
\begin{equation}\label{ast}
\begin{cases}u_t=v_x\\v_t=\frac{u_x}{u^2}-2x\end{cases},
\end{equation}
that is clearly a non-homogeneous  system
\begin{equation*}
V=\begin{pmatrix}
0&1\\\frac{1}{u^2}&0
\end{pmatrix}
\qquad \text{,} \qquad
W=\begin{pmatrix}
0\\-2x
\end{pmatrix}
\end{equation*}
The previous system admits a Hamiltonian formulation with the Dubrovin-Novikov operator
\begin{equation}
Q^{ij}=\begin{pmatrix}
0&1\\1&0
\end{pmatrix}\partial_x.
\end{equation}
The reader can check that (1--5) are satisfied via symbolic computations. 
\end{example}
\section{Non-local operators}In \cite{FerMok1}, E. V. Ferapontov and O. I. Mokhov proposed a generalization of first order homogeneous Hamiltonian operators adding a tail of nonlocal terms to the Dubrovin-Novikov operators. This extension is still releated with hydrodynamic type systems if the Hamiltonian density $h$ is of hydrodynamic type $h=h(\mathbf{u})$.

\noindent Now, for non-homogeneous quasilinear systems
 the following equality is verified
\begin{equation}
  \label{eq:9}
  \langle\ell_F(\varphi),\psi\rangle - \langle\varphi,\ell_F^*(\psi)\rangle =
  \sum_{i=1}^nD_i(a^i).
\end{equation}
If $\varphi$ is a symmetry, the first summand at the left-hand side vanishes.
Moreover, if we lift the remaining identity on the cotangent covering
$\ell^*_F(\mathbf{p})=0$, we have a conservation law on the right-hand
side. Dually, we could use a cosymmetry and lift the remaining terms on the
tangent covering. Let us  find an explicit
formula for the conservation law:
\begin{align}
  \label{eq:12}
  (\partial_t\varphi^i -( V^i_{j,k}u^j_x+W^i_{,k})\varphi^k - & V^i_j \partial_x\varphi^j)\psi_i\notag
  \\
  & - \varphi^i( - \partial_t\psi_i + (V^k_{i,j}u^j_x - V^k_{j,i}u^j_x-W^k_{,i}) \psi_k + V^k_i
  \partial_x\psi_k)\notag
  \\
  = &
  \partial_t(\varphi^i\psi_i) - \partial_x(V^i_j\varphi^j\psi_i)
\end{align}
Indeed, $-W^i_{,k}\varphi^k\psi_i+\varphi^iW^k_{,i}\psi_k=0$.

\noindent Therefore, we denote with $r$ the new nonlocal variable on the cotangent covering corresponding with each
symmetry $\varphi$ 
\begin{equation}
  \label{eq:189}
  r_t = V^i_j\varphi^jp_i,\qquad r_x = \varphi^ip_i,
\end{equation}
Nonlocal variables of such type were introduced by P. Kersten, I. S. Krasil'schik and A. Verbovetsky \cite{KerstenKrasilshchikVerbovetsky:GSDBTEq}.

\vspace{3mm}

\noindent As an example, let us firstly focus on the following type of operators
\begin{equation}\label{fmbra}
B^{ij}=g^{ij}\partial_x+b^{ij}_ku^k_x+cu^i_x\partial_x^{-1}u^j_x,
\end{equation} 
where $c$ is an arbitrary constant. Operators of this kind are associated to nonlocal Poisson brackets also known as Ferapontov-Mokhov brackets. Extending this definition to operators, then \eqref{fmbra} we call operators of this type \emph{Ferapontov-Mokhov operators}. In \cite{MokFer12}, the authors proved the Hamiltonianity conditions for the nondegenerate case:
\begin{theorem} If $\det g^{ij}\neq 0$, the Ferapontov-Mokhov operator is Hamiltonian if and only if $g^{ij}$ is a metric of constant curvature $c$ and $\Gamma^{ij}_k$ are Christoffel symbols of third kind compatible with $g^{ij}$.
\end{theorem}

\noindent In order to find compatibility conditions between quasilinear systems and Ferapontov-Mokhov operators we briefly recall that in \cite{VerVit1} the authors proved the following result for homogeneous systems
\begin{theorem}\label{sec:non-local-operators-1}
  Let us consider a non-local first order  operator\begin{equation}
  B^{ij}=g^{ij}\partial_x+\Gamma^{ij}_ku^k_x+w^i_su^s_x\partial_x^{-1}w^j_lu^l_x,\end{equation} whose
  non-local part is defined by a hydrodynamic type symmetry
  $\varphi^i = w^i_ju^j_x$, and the hydrodynamic type
  system \begin{equation}\label{eq11}
  u^i_t=V^i_j(u)u^j_x.
  \end{equation} Then, the compatibility conditions
  $\ell_F(B(\mathbf{p}))=0$ for the operator $B$ to be a Hamiltonian operator
  for the hydrodynamic type system~\eqref{eq11} are equivalent to the following
  system:
\begin{enumerate}
\item $g^{ik}V^j_k=g^{jk}V^i_k$,
\item $\nabla^iV^j_k=\nabla^jV^i_k$.
\end{enumerate}\end{theorem}
As Corollary we obtain that when $\varphi^i=\delta^i_ju^j_x$ (that is when $w^i_j=\delta^i_j$) the Ferapontov-Mokhov operators are compatible with a homogeneous hydrodynamic type system if Tsarev's conditions are satisfied and $\varphi^i=u^i_x$ is a symmetry for the system. Note that $\varphi^i=u^i_x$ is a symmetry for the system if and only if $V^i_j$ does not depend on $x$.  As an example, let us consider the Chaplygin equation: 
\begin{example}The Chaplygin gas equation is given by the following system
\begin{equation}\label{syscha}
\begin{cases}u_t=uu_x+\frac{1}{v^3}v_x\\v_t=vu_x+uv_x
\end{cases}.
\end{equation}Applying to  \eqref{syscha} the change of variables $u=\tilde{u}-\tilde{v}^{-1}$ and $v=\tilde{u}+\tilde{v}^{-1}$, the system can be written (see \cite{mokhov98:_sympl_poiss}) in diagonal form:
\begin{equation*}
\begin{cases}u_t=vu_x\\v_t=uv_x
\end{cases}
\end{equation*} \noindent One can easily check that the previous theorem is satisfied for the structure \eqref{fmbra} given by the contravariant metric
\begin{equation}
g^{ij}=\begin{pmatrix}-[(c_1+k)+c_2u+c_3u^2](u-v)^2&0\\
0&[c_1+c_2v+c_3v^2](u-v)^2\end{pmatrix}
\end{equation} whose constant Riemannian curvature is $k$. 
\end{example}
\begin{remark} We can generalize the previous result for non-homogeneous quasilinear systems of first order PDEs. We associate to \eqref{fmbra} the linear vector function \begin{equation}
B^i(\mathbf{p,r})=g^{ij}p_{j,x}+\Gamma^{ij}_ku^k_xp_j+cu^i_xr,
\end{equation}
and obtain that given a non-homogeneous quasilinear system
\begin{equation}\label{oop}
F^i: \quad u^i_t=V^i_j(u)u^j_x+W^i(u)
\end{equation} and an Hamiltonian operator \eqref{fmbra},  the following conditions are equivalent:
\begin{itemize}
\item[(i)] $\ell_F(B(\mathbf{p,r}))=0$;
\item[(ii)]\begin{enumerate}\item $\nabla^iV^j_k=\nabla^jV^i_k$;
\item $g^{ik}V^j_k=g^{jk}V^i_k$;
\item $\nabla^iW^j+\nabla^jW^i=0$;
\item $\nabla^iW^j_{,x}=0$;
\item $\nabla_k\nabla^iW^j=0$.
\end{enumerate}
\end{itemize}\end{remark}


\noindent

\subsection{First-order  operators extended by isometries}
Finally, let us focus on particular nonlocal  operators 
\begin{equation}\label{3}
B^{ij}=g^{ij}\partial_x+\Gamma^{ij}_ku^k_x+\alpha f^i\partial_x^{-1}f^j,
\end{equation}
where $\alpha$ is a real constant and $f=f^i(\textbf{u})\partial_{u^i}$ is a vector field. Such operators were introduced by E. V. Ferapontov in \cite{ferapontov92:_non_hamil}, who also proved the conditions of hamiltonianity:
\begin{theorem}[\cite{ferapontov92:_non_hamil}]\label{thm1}
Necessary and sufficient conditions for the Poisson bracket determinated by \eqref{3} to be Hamiltonian are that
\begin{equation}
B_0^{ij}=g^{ij}\partial_x+\Gamma^{ij}_ku^k_x
\end{equation}is Hamiltonian and $f$ satisfies
\begin{enumerate}
\item $\nabla^if^j+\nabla^jf^i=0$ and
\item $f^k\nabla^if^j+<\text{cyclic}>=0$.
\end{enumerate}where $\nabla^j=g^{js}\nabla_s$ and $\nabla_s$ is the Levi-Civita connection defined by $g_{ij}$.
\end{theorem}
In particular, $f$ must be an infinitesimal isometry for $g$. 

\noindent These kind of operators are largely present in Hamiltonian formalism for PDEs: Ferapontov presented three examples (the Nonlinear Schroedinger equation, the Heisenberg  magnet and the Landau-Lifshits equations).  Moreover, a role of such operators in integrable systems was introduced by M. V. Pavlov, R. Vitolo and the author in \cite{PavVerVit1}. In the paper, the authors presented an  example of non-homogeneous hydrodynamic type system in three components admitting a bi-hamiltonian structure. Such a structure is given by a Dubrovin-Novikov operator and one of type \eqref{3}.

\noindent Let us consider  nonlocal operators  \eqref{3} with $\alpha=1$. The operator must satisfy the conditions presented in Theorem \eqref{thm1} and it can be 
identified with the linear vector function:
\begin{equation}
  \label{eq:6}
  B^i(\mathbf{p,r}) = g^{ij}p_{j,x} + \Gamma^{ij}_k u^k_x p_j + f^i r,
\end{equation}
where $r$ is the nonlocal variable defined by $r_x = f^jp_j$ and $r_t = V^i_jf^jp_i$. Then, for the homogeneous case $W^i(x,\textbf{u})=0$, the following result is proved.

\begin{lemma}\label{thm3} For the Hamiltonian operator $B^{ij}$ in \eqref{3} and the homogeneous hydrodynamic type system 
\begin{equation}\label{sys}
u^i_t=V^i_j(u)u^j_x,
\end{equation}
the following are equivalent:
\begin{itemize}
\item[(i)]$\ell_F(B(\mathbf{p,r}))=0$;
\item[(ii)]\begin{enumerate}
\item $f$ is a symmetry for the system \eqref{sys};
\item $\nabla^iV^j_k=\nabla^jV^i_k;$
\item $g^{ik}V^j_k=g^{jk}V^i_k;$
\item $f^kV^j_kf^i=f^jV^i_kf^k$.\end{enumerate}
\end{itemize}
\end{lemma}
\begin{proof}
\noindent We first notice that conditions (2) and (3) of the system imply $\ell_F(A)=0$, as necessary and sufficient conditions of existence of the Hamiltonian structure (see Theorem 5,  \cite{VerVit1}). Therefore, 
\begin{align}\begin{split}
\ell_F(\tilde{A})&=\ell_F(A)+\ell_F(f^ir)\\
&=0+\partial_t(f^ir)-V^i_{j,k}u^j_xf^kr-V^i_j\partial_x(f^jr)\\
&=\partial_t(f^i)r+r_tf^i-V^i_{j,k}u^j_xf^kr-V^i_j\partial_x(f^i)r-V^i_jf^jr_x\end{split}
\end{align}
By substituing $r_x=f^kp_k$ and $r_t=V^k_jf^jp_k$
\begin{align}
\begin{split}&\partial_t(f^i)r+(V^k_jf^jp_k)f^i-V^i_{j,k}u^j_xf^kr-V^i_j\partial_x(f^i)r-V^i_jf^j(f^kp_k)\\
&=(\partial_t(f^i)-f^kV^i_{j,k}u^j_x-V^i_j\partial_x(f^j))r+(V^k_jf^jf^k-V^i_jf^jf^k)p_k\\
&=\ell_F(f)+(f^jV^k_jf^i-f^jV^i_jf^k)p_k\\
&=0.
\end{split}
\end{align}

\noindent Viceversa,  by computing the linearization of the operator $\tilde{A}$ we obtain:
\begin{align}\ell_F(\tilde{A})=
 \begin{split}  &\label{ag0}( - V^i_kg^{kj} + V^j_k g^{ki})p_{j,xx} \end{split}+
  \\
  \begin{split}\label{ag1}
 & \Big( g^{ij}_kV^k_l u^l_x + g^{ik}(V^j_{k,m}u^m_x - V^j_{m,k}u^m_x)
    + g^{ik}V^j_{k,m}u^m_x+ \Gamma^{ik}_h u^h_xV^j_k
    \\
    &\hphantom{ciao} - V^i_{l,k}u^l_xg^{kj}  - V^i_kg^{kj}_h u^h_x
    - V^i_k\Gamma^{kj}_h u^h_x\Big) p_{j,x} +
  \end{split}
  \\
  \begin{split}\label{ag2}
  & \Big( g^{ik}(V^j_{k,ml}u^l_xu^m_x+V^j_{k,m}u^m_{xx} - V^j_{m,kl}u^l_xu^m_x-V^j_{m,k}u^m_{xx})\\
  &\hphantom{ciao}
      + \Gamma^{ij}_{k,h}V^h_l u^l_x u^k_x +\Gamma^{ij}_{k}V^k_{l,m}u^m_x u^l_x+\Gamma^{ij}_kV^k_lu^l_{xx}+
  \Gamma^{ik}_l u^l_x(V^j_{k,h}u^h_x \\&\hphantom{ciao}- V^j_{h,k}u^h_x) 
  - V^i_{l,k}u^l_x\Gamma^{kj}_h u^h_x  
  - V^i_k( \Gamma^{kj}_{h,l}u^l_x u^h_x + \\&\hphantom{ciao}\Gamma^{kj}_h u^h_{xx})+(V^j_hf^hf^i-V^i_hf^hf^j)\Big)p_j+\end{split}
\\\begin{split}\label{ag3}
  &  \hphantom{ciao}\Big(\partial_t(f^i)-f^kV^i_{j,k}u^j_x-V^i_j\partial_x(f^j)\Big)r.\end{split}
\end{align}
Now we annihilate coefficients of polynomials in $p_{j,x}$, $p_{j,xx}$ and $p_j$. The coefficients of $p_{j,xx}$ and $p_{j,x}$ in \eqref{ag0} and \eqref{ag1} are the same as in Tsarev Theorem of compatibility  \cite{VerVit1}. Moreover, by annihilating the coefficient of the variable $r$ in \eqref{ag3} we notice that $f$ must be a symmetry of the system \eqref{sys}. The coefficient of $p_{j}$ in \eqref{ag2}  is:
\begin{align}\begin{split}
& g^{ik}(V^j_{k,ml}u^l_xu^m_x+V^j_{k,m}u^m_{xx} - V^j_{m,kl}u^l_xu^m_x-V^j_{m,k}u^m_{xx})
      + \Gamma^{ij}_{k,h}V^h_l u^l_x u^k_x + 
  \\ &\hphantom{ciao} +\Gamma^{ij}_{k}V^k_{l,m}u^m_x u^l_x+\Gamma^{ij}_kV^k_lu^l_{xx}+
  \\
  & \hphantom{ciao}
  \Gamma^{ik}_l u^l_x(V^j_{k,h}u^h_x - V^j_{h,k}u^h_x) 
  - V^i_{l,k}u^l_x\Gamma^{kj}_h u^h_x  
   \\
  & \hphantom{ciao}
  - V^i_k( \Gamma^{kj}_{h,l}u^l_x u^h_x + \Gamma^{kj}_h u^h_{xx})+
  \\
  & \hphantom{ciao}
  +V^j_hf^hf^i-V^i_hf^hf^j.\end{split}
\end{align}
Therefore, condition (4) is satisfied: \begin{equation}V^k_jf^jf^i-V^i_jf^jf^k=0.\end{equation}
Finally, the coefficients of $u^l_xu^h_x$ and $u_{xx}^l$ are just differential consequences of the previous conditions. Then, conditions (2) and (3) are satisfied.
\end{proof}

\vspace{5mm}

\noindent For non-homogeneous hydrodynamic type systems the following is valid

\begin{theorem}\label{thm1nonloc}
Let $B^{ij}$ be an operator of type \eqref{3}, satisfying the hamiltonianity conditions in Theorem \ref{thm3} and 
\begin{equation}\label{sys11}
u^i_t=V^i_j(u)u^j_x+W^i(x,{u})
\end{equation}
a non-homogeneous quasilinear system, then the following conditions are equivalent:
\begin{itemize}
\item[(i)]$\ell_F(\tilde{A})=0$;
\item[(ii)]\begin{enumerate}
\item $f$ is a symmetry for the system \eqref{sys11};
\item $\nabla^iV^j_k=\nabla^jV^i_k$;
\item $g^{ik}V^j_k=g^{jk}V^i_k$;
\item $\nabla^iW^j_x-f^kV^j_kf^i+f^kV^i_kf^j=0$;
\item $\nabla^iW^j+\nabla^jW^i=0$;
\item $\nabla_k\nabla^iW^j=0$.\end{enumerate}
\end{itemize}
\end{theorem}
\begin{proof}
We notice that the Theorem is proved for $f^i=0$  by the previous considerations and due to Lemma \ref{thm3} . 
Moreover, it is easy to observe that by linearity of $\ell_F$ we have
\begin{equation*}
\ell_F(\tilde{A}^i)=\ell_F(A^i)+\ell_F(f^ir),
\end{equation*}
where \begin{equation}
A^i=g^{ij}p_{j,x}+\Gamma^{ij}_ku^k_xp_j,
\end{equation}
and \begin{equation}
\ell_F(f^ir)=\ell_F(f)r+f^ir_t+V^i_jf^jr_x.
\end{equation}
By using \eqref{eq:189} and the previous linearization, it is possible to observe that the only different term to be annihilate is the coefficient of $p_j$. In particular, we have the condition
\begin{equation}
-\nabla^iW^j_x+f^kV^j_kf^i-f^kV^i_kf^j=0,
\end{equation}
and so $(i)$ is verified if and only if $(ii)$ is. 
\end{proof}
\begin{remark}We remark that both the nonlocal cases studied in this section strictly link the operators with the system. Indeed, the compatibility conditions in Theorems \ref{sec:non-local-operators-1} and \ref{thm1nonloc}  remark the necessity for the nonlocal extension to be a symmetry for the system.  \end{remark}

\begin{corollary}If $\alpha\in\mathbb{R}$ and \begin{equation}
B^i(\mathbf{p,r})=g^{ij}p_{j,x}+\Gamma^{ij}_ku^k_xp_j+\alpha f^ir
\end{equation} then condition (4) of Theorem \ref{thm1nonloc} is substitued by the following
\begin{equation}
\nabla^iW^j_x+\alpha f^kV^j_kf^i-\alpha f^kV^i_kf^j=0.
\end{equation} \end{corollary}
An example of non-homogeneous hydrodynamic type systems admitting Ferapontov operators is given by the constant astigmatism equation
\begin{example}
Consider the system \eqref{ast} associated to the constant astigmatism equation 
and define the operator:
\begin{equation}\label{co}
P^{ij}=g^{ij}\partial_x+\Gamma^{ij}_ku^k_x+2f^i\partial_x^{-1}f^j
\end{equation}
where \begin{equation}
g^{ij}=\begin{pmatrix}
2u&0\\0&\frac{2}{u}
\end{pmatrix} \qquad , \qquad f=\begin{pmatrix}
0\\1
\end{pmatrix}
\end{equation} and $\Gamma^{ij}_k$ are Christoffel symbols of third type of the metric $g^{ij}$. System \eqref{ast} is Hamiltonian with respect to the operator $P^{ij}$. 

\noindent Here, $f$ is a symmetry of system \eqref{ast} and it is easy to verify that for $\alpha=2$  the conditions in Theorem \ref{thm1nonloc} are satisfied.
\end{example}
The existence of a bi-Hamiltonian structure for the constant astigmatism equation was investigated by Pavlov and Zykov in \cite{PavZyk}. Then, the constant astigmatism equation has two Hamiltonian structures. Moreover, the authors proved that the operators $Q^{ij}$ in Example \ref{examasti} and $P^{ij}$ in \eqref{co} are compatible, i.e. the system is integrable \cite{Magri:SMInHEq}. 

\vspace{3mm}

\noindent Finally, the following is an example in dimension 3. 

\begin{example} \label{exam1}
Let us consider the non-homogeneous quasilinear system
\begin{equation}
\begin{cases}
u_{t}=-\frac{3v^{2}}{2w^{2}}v_{x}+\frac{v^{3}}{w^{3}}w_{x}-x \\
v_{t}=u_{x}+\frac{3v}{w}v_{x}-\frac{3v^{2}}{2w^{2}}w_{x} \\
w_{t}=v_{x}
\end{cases}.
\label{eq:20}
\end{equation}
This system firstly appeared in \cite{PavVerVit1} and possesses a bi-Hamiltonian structure 
\begin{equation*}
{A}^{ij}\frac{\delta H_1}{\delta u^j}=B^{ij}\frac{\delta H_2}{\delta u^j},
\end{equation*}
where \begin{equation}A^{ij}=\begin{pmatrix}
0&0&1\\0&1&0\\1&0&0
\end{pmatrix}\partial_x\end{equation}
and \begin{equation*}B^{ij}=g^{ij}\partial_x+\Gamma^{ij}_ku^k_x+f^i\partial_x^{-1}f^j\end{equation*}
is a nonlocal operator of type \eqref{3} with
\begin{equation}
g^{ij}=%
\begin{pmatrix}
\frac{v^{3}}{w^{2}} & \frac{-3v^{2}}{2w} & -v+1 \\ 
\frac{-3v^{2}}{2w} & 2v+1 & w \\ 
-v+1 & w & 0%
\end{pmatrix}\quad \text{and}\quad f=\partial_u %
  \label{12}.
\end{equation}
 \noindent This is a non-trivial example of integrable non-homogeneous hydrodynamic type system, where $W^i=(-x,0,0)^{\text{T}}$. 
As expected, the system is compatible with the non-local operator $B^{ij}$ satisfying the conditions of Theorem \ref{thm1nonloc}. Moreover, the operator $A^{ij}$ is compatible with the system in sense of Theorem \ref{thmloc}.
\end{example}



\section{Conclusions}
The work here presented follows the approach established in \cite{VerVit1} of cotangent and tangent coverings for compatibility conditions, with a new application. The results cover a larger class of operators  then in \cite{VerVit1} and are confirmed by some known examples of PDEs. In particular, the study of nonlocal operators extended by isometries and quasilinear systems turns out to have a significant role in bi-Hamiltonian structures and integrability. However, a deeper study of systems which admit such Hamiltonian operators is still needed. 

\noindent The approach here presented emphasises the possibility to study a very large class of operators and systems, with different applications. In particular, as a future perspective, it would be interesting to find similar conditions for non-homogeneous operators composed by the sum of homogeneous ones. Examples of these kind appear very often in literature (e.g. the KdV equation in \cite{mokhov98:_sympl_poiss}, or the AKNS equation in \cite{falq}). The starting point will be the study of compatibility conditions for quasilinear systems of first order PDEs and non-homogeneous operators composed by a first-order operator plus an operator of order zero:
\begin{equation}
C^{ij}=g^{ij}\partial_x+\Gamma^{ij}_ku^k_x+\omega^{ij}
\end{equation}


\noindent Finally, a possible application of this method can be introduced for symplectic operators and the Monge-Amp\`ere equations. 

\vspace{5mm}

\textbf{Acknowledgements.} The author thanks R. Vitolo, E. Ferapontov, M. Pavlov and M. Menale for stimulating discussions. The author also acknowledges the financial support of GNFM of the Istituto Nazionale di Alta Matematica and of PRIN 2017 \textquotedblleft Multiscale phenomena in Continuum
Mechanics: singular limits, off-equilibrium and transitions\textquotedblright,
project number 2017YBKNCE.


\begin{thebibliography}{10}

\bibitem{AgaFer}
S. I . Agafonov and E. V. Ferapontov,
\newblock \emph{Systems of conservation laws in the context of the projective theory of congruences},
\newblock Izv. RAN. Ser. Mat., 1996, Volume 60, Issue 6, Pages 3–30.

\bibitem{Vit2}
M. Casati, P. Lorenzoni, D. Valeri, R. Vitolo,
\newblock  \emph{Weakly nonlocal 
Poisson brackets: tools, examples, computations}, 
\newblock Computer Physics Communications (2022), \url{arXiv:2101.06467}.



\bibitem{DubKri} B. A. Dubrovin, I. M. Krichever, and S. P. Novikov.
\newblock  Integrable systems. I. In \emph{Dynamical Systems IV}, volume 4 of Encyclopaedia of Mathematical Sciences, 
\newblock pages 173–280. Springer-Verlag, Berlin, 2 edition, 2001.

\bibitem{DN83}
B.A. Dubrovin and S.P. Novikov.
\newblock \emph{Hamiltonian formalism of one-dimensional systems of hydrodynamic type
  and the Bogolyubov--Whitham averaging method}.
\newblock {\em Soviet Math. Dokl.}, 27(3):665--669, 1983.

\bibitem{DubrovinNovikov:PBHT}
B.~A. Dubrovin and S.~P. Novikov.
\newblock \emph{Poisson brackets of hydrodynamic type}.
\newblock {\em Soviet Math. Dokl.}, 30:651--654, 1984.



\bibitem{falq}
G. Falqui, 
\newblock \emph{On a Camassa-Holm type equation with two dependent variables}, 
\newblock J. Phys. A.: Math. Gen. 39(2006),327-342.

\bibitem{F95:_nl_ho}
E.V. Ferapontov.
\newblock \emph{Nonlocal {H}amiltonian operators of hydrodynamic type: Differential
  geometry and applications}.
\newblock {\em Amer. Math. Soc. Transl.}, 170(2):33--58, 1995.

\bibitem{ferapontov92:_non_hamil} E.V. Ferapontov. \newblock \emph{Non local
matrix hamiltonian operators, differential geometry, and applications}. %
\newblock {\em Theoret. and Math. Phys.}, 91(3):642--649, 1992.



\bibitem{ferapontov14:_hamil}
E.V. Ferapontov, P.~Lorenzoni, and A.~Savoldi.
\newblock \emph{Hamiltonian operators of {D}ubrovin--{N}ovikov type in $2d$}.
\newblock {\em Lett. Math. Phys.}, 105(3):341--377, 2014.
\newblock arXiv:1312.0475.

\bibitem{MokFer12}
E. V. Ferapontov and O. I. Mokhov,
\newblock \emph{Non-local Hamiltonian operators of hydrodynamic type related to metrics of constant curvature}. Russian Math. Surveys, 1990, V. 45, n.3, p. 218-219.

\bibitem{FPV14}
E.V. Ferapontov, M.V. Pavlov, and R.F. Vitolo.
\newblock \emph{Projective-geometric aspects of homogeneous third-order {H}amiltonian
  operators}.
\newblock {\em J. Geom. Phys.}, 85:16--28, 2014.
\newblock \texttt{DOI:10.1016/j.geomphys.2014.05.027}.

\bibitem{FerPavVit1} E. V. Ferapontov, M. V. Pavlov and R. Vitolo, 
\newblock \emph{Systems of conservation laws with third-order Hamiltonian structures}, 
\newblock Lett. Math. Phys. 108, Issue 6 (2018), 1525-1550.





\bibitem{KerstenKrasilshchikVerbovetsky:HOpC}
P.~Kersten, I.~Krasil'shchik, and A.~Verbovetsky.
\newblock \emph{Hamiltonian operators and $\ell^*$-coverings}.
\newblock {\em J. Geom. Phys.}, 50:273--302, 2004.

\bibitem{KerstenKrasilshchikVerbovetsky:GSDBTEq}
P.~Kersten, I.~Krasil'shchik, and A.~Verbovetsky.
\newblock \emph{A geometric study of the dispersionless {B}oussinesq type equation}.
\newblock {\em Acta Appl. Math.}, 90:143--178, 2006.


\bibitem{krasilshchik84:_nonloc}
I.S. Krasil'shchik and A.M. Vinogradov.
\newblock \emph{Nonlocal symmetries and the theory of coverings: {A}n addendum to
  {A}.{M}. {V}inogradov's `local symmetries and conservation laws'}.
\newblock {\em Acta Appl. Math.}, 2:79--96, 1984.

\bibitem{KVV17}
J.~Krasil'shchik, A.~Verbovetsky, and R.~Vitolo.
\newblock {\em The symbolic computation of integrability structures for partial
  differential equations}.
\newblock Texts and Monographs in Symbolic Computation. Springer, 2018.
\newblock ISBN 978-3-319-71654-1; see \url{http://gdeq.org/Symbolic_Book} for
  downloading program files that are discussed in the book.


\bibitem{LSV:bi_hamil_kdv}
P.~Lorenzoni, A.~Savoldi, and R.~Vitolo.
\newblock \emph{{B}i-{H}amiltonian systems of {KdV} type}.
\newblock {\em J. Phys. A}, 51(4):045202, 2018.

\bibitem{Magri:SMInHEq}
F.~Magri.
\newblock \emph{A simple model of the integrable {H}amiltonian equation}.
\newblock {\em J. Math. Phys.}, 19:1156--1162, 1978.

\bibitem{FerMok1}
O. I. Mokhov, E. V. Ferapontov, 
\newblock \emph{Hamiltonian Pairs Associated with Skew-Symmetric Killing Tensors on Spaces of Constant Curvature}.
\newblock Funktsional. Anal. i Prilozhen., 28:2 (1994), 60–63; Funct. 
Anal. Appl., 28:2 (1994), 123–125.

\bibitem{mokhov98:_sympl_poiss}
O.I. Mokhov.
\newblock \emph{Symplectic and {P}oisson geometry on loop spaces of smooth manifolds
  and integrable equations}.
\newblock In S.P. Novikov and I.M. Krichever, editors, {\em Reviews in
  mathematics and mathematical physics}, volume~11, pages 1--128. Harwood
  academic publishers, 1998.



\bibitem{NovikovManakovPitaevskiiZakharov:TS}
S.~P. Novikov, S.~V. Manakov, L.~P. Pitaevskii, and V.~E. Zakharov.
\newblock {\em Theory of Solitons}.
\newblock Plenum Press, 1984.


\bibitem{Olver:ApLGDEq}
P.J. Olver.
\newblock {\em Applications of {L}ie Groups to Differential Equations}.
\newblock Springer-Verlag, 2nd edition, 1993.

\bibitem{PavVerVit1}
M. V. Pavlov, P. Vergallo, R. Vitolo.
\newblock \emph{Classification of bi-Hamiltonian pairs extended by isometries},
Proc. Roy. Soc. A, June 2021.

\bibitem{PavVit1}
M.V. Pavlov, R.F. Vitolo.
\newblock \emph{Bi-Hamiltonian structure of the Oriented Associativity equation} , 
\newblock J. Phys. A: Theor. Math. - Letters, Volume 52, Number 20 (2019).

\bibitem{PavZyk} M.V. Pavlov and S.A. Zykov. \newblock \emph{{L}%
agrangian and {H}amiltonian structures for the constant astigmatism
equation}. \newblock {\em J. Phys. A}, 46:395203, 2013. \newblock Ar{X}iv: 
\texttt{1212.6239}.

\bibitem{tsarev85:_poiss_hamil}
S.P.~Tsarev.
\newblock \emph{On {P}oisson brackets and one-dimensional {H}amiltonian systems of
  hydrodynamic type}.
\newblock {\em Soviet Math. Dokl.}, 31(3):488--491, 1985.

\bibitem{tsarev91:_hamil}
S.P. Tsarev.
\newblock \emph{The geometry of {H}amiltonian systems of hydrodynamic type. the
  generalized hodograph method}.
\newblock {\em Math. USSR-Izv.}, 37(2):397--419, 1991.

\bibitem{VerVit1}
P. Vergallo and R. Vitolo,
\newblock \emph{Homogeneous Hamiltonian operators and the theory of coverings}
\newblock {\em Diff. Geom. Its Appl.}, April 2021.
 

\bibitem{VerVit2}
P. Vergallo and R. Vitolo,
\newblock \emph{Projective geometry of homogeneous second order Hamiltonian operators},
\newblock {\em arXiv}, \url{https://arxiv.org/abs/2203.04237}, March 2022.


\bibitem{krasilshchikvinogradov84}
A.M. Vinogradov and I.S. Krasil'shchik.
\newblock \emph{On the theory of nonlocal symmetries of nonlinear partial
  differential equations}.
\newblock {\em Soviet Math. Dokl.}, 29:337--341, 1984.

\bibitem{Vit1}
R. Vitolo.
\newblock \emph{Computing with Hamiltonian operators}, 
\newblock Computer Physics 
Communications Volume 244 (2019), 228-245.





\end{thebibliography}
\end{document}